\documentclass[11pt]{article}
\usepackage{fullpage}

\usepackage{graphics}
\usepackage[dvips]{epsfig}

\usepackage{amsmath}
\usepackage{amssymb}
\usepackage{amsfonts}
\usepackage{graphicx}

\newtheorem{theorem}{Theorem}[section]
\newtheorem{lemma}{Lemma}[section]
\newtheorem{corollary}{Corollary}[section]

\newcommand{\qed}{\hfill $\Box$ \bigbreak}
\newenvironment{proof}{\noindent {\bf Proof.}}{\qed}

\newcommand{\remove}[1]{}

\begin{document}

\baselineskip  0.2in %  0.2in %0.18in si on veut compact
\parskip     0.05in %    0.1in % 0.0in  pour compacter
\parindent   0.0in %    0.0in % 0.3in pour voir les paragraphes

\title{{\bf Deterministic rendezvous with detection using beeps}
\footnote{A preliminary version of this paper appeared in
Proc. 11th International Symposium on Algorithms and Experiments for Wireless Sensor Networks (ALGOSENSORS 2015), LNCS 9536, 85-97. }}

\author{
Samir Elouasbi\thanks{D\'{e}partement d'informatique, Universit\'{e} du Qu\'{e}bec en Outaouais,
Gatineau, Qu\'{e}bec J8X 3X7,
Canada. E-mail: elos02@uqo.ca.}
\and
Andrzej Pelc\thanks{D\'{e}partement d'informatique, Universit\'{e} du Qu\'{e}bec en Outaouais,
Gatineau, Qu\'{e}bec J8X 3X7,
Canada. E-mail: pelc@uqo.ca.
Supported in part by NSERC discovery grant 8136 -- 2013 
and by the Research Chair in Distributed Computing of
the Universit\'{e} du Qu\'{e}bec en Outaouais.}
}

\date{ }
\maketitle

\begin{abstract}

Two mobile agents, 
starting at arbitrary, possibly different times from arbitrary nodes of an unknown network, have to meet at some node.
Agents move in synchronous rounds: in each round an agent can either stay at the current node or move to one of its neighbors.
Agents have different labels which are positive integers. Each agent knows its own label, but not the label of the other agent.
In traditional formulations of the rendezvous problem, meeting is accomplished when the agents get to the same node in the same round.
We want to achieve a more demanding goal, called {\em rendezvous with detection}: agents must become aware that the meeting is accomplished,
simultaneously declare this and stop. This awareness depends on how an agent can communicate to the other agent its presence at a node.
We use two variations of the arguably weakest model of communication, called the {\em beeping model}, introduced in \cite{CK}. In each round an agent can either listen or beep. In the {\em local beeping model},
an agent hears a beep in a round if it listens in this round and if the other agent is at the same node and beeps. In the {\em global beeping model},
an agent hears a {\em loud} beep in a round if it listens in this round and if the other agent is at the same node and beeps, and it hears a {\em soft}
beep in a round if it listens in this round and if the other agent is at some other node and beeps.

We first present a deterministic algorithm of rendezvous with detection working, even for the local beeping model, in an arbitrary unknown network in time polynomial in the size of the network
and in the length of the smaller label (i.e., in the logarithm of this label). However, in this algorithm, agents spend a lot of energy: the number of moves
that an agent must make, is proportional to the time of rendezvous. It is thus natural to ask if {\em bounded-energy agents}, i.e., agents that can make at most 
$c$ moves, for some integer $c$, can always achieve rendezvous with detection as well. This is impossible for some networks of unbounded size. Hence we rephrase the question: Can bounded-energy agents always achieve rendezvous with detection in bounded-size networks?
We prove that the answer to this question is positive, even in the local beeping model but, perhaps surprisingly, this ability comes at a steep price of time: the meeting time of bounded-energy agents is {\em exponentially} larger
than that of unrestricted agents. By contrast, we show an algorithm for rendezvous with detection in the global beeping model that works for bounded-energy agents
(in bounded-size networks) as fast as for unrestricted agents.

\vspace{2ex}

\end{abstract}

\vspace{2ex}

\vfill

\thispagestyle{empty}
\setcounter{page}{0}
\pagebreak

%%%%%%%%%%%%%%%%%%%%%%%%%%%%%%%%%%%%%%%%%%%%%%%%%%%%%%%%%%%
\section{Introduction}
%%%%%%%%%%%%%%%%%%%%%%%%%%%%%%%%%%%%%%%%%%%%%%%%%%%%%%%%%%%
{\bf The background and the problem.}
Two mobile agents, 
starting at arbitrary, possibly different times from arbitrary nodes of an unknown network, have to meet at some node of it.
This task is known as rendezvous ~\cite{alpern02b}.
The network is modeled as a simple undirected connected graph,  and agents move in synchronous rounds: 
in each round an agent can either stay at the current node or move to one of its neighbors. 
Hence in each round an agent is at a specific node.
Agents are mobile entities with unlimited memory; from the computational point of view they are modeled as 
Turing machines. In applications, these entities may represent mobile robots navigating in a labyrinth or in corridors of a building, or software agents moving in a communication network.
The purpose of meeting might be to exchange data previously collected by the agents at nodes of the network,
or to coordinate future network maintenance tasks, for example checking functionality of websites or of sensors connected in a network. 

Agents have different labels which are positive integers. Each agent knows its own label, but not the label of the other agent.
Agents do not know the topology of the network, they do not have any bound on its size. They do not know the starting node or activation time of the other agent. They cannot mark the visited nodes in any way.
Each agent appears at its starting node at the time of its activation by the adversary. 

We seek rendezvous algorithms that do not
rely on the knowledge of node labels, and can work in anonymous networks as well  (cf. \cite{alpern02b}). 
The importance of designing such algorithms
is motivated by the fact that, even when nodes are equipped with distinct labels, agents may be unable to perceive them
because of limited sensory capabilities, 
or nodes may refuse to reveal their labels to agents, e.g., due to security or privacy reasons.
On the other hand, we assume that
edges incident to a node $v$ have distinct labels in 
$\{0,\dots,d-1\}$, where $d$ is the degree of $v$. Thus every undirected
edge $\{u,v\}$ has two labels, which are called its {\em port numbers} at $u$
and at $v$. Port numbering is {\em local}, i.e., there is no relation between
port numbers at $u$ and at $v$. An agent entering a node learns the port of entry and the degree of the node. 
Note that, in the absence of port numbers, rendezvous is usually impossible, 
as all ports at a node look identical to an agent and the adversary may prevent the agent from taking some edge incident to the current node.

In traditional formulations of the rendezvous problem, meeting is accomplished when the agents get to the same node in the same round.
We want to achieve a more demanding goal, called {\em rendezvous with detection}: agents must become aware that the meeting is accomplished,
simultaneously declare this and stop. This awareness depends on how an agent can communicate to the other agent its presence at a node.
We use two variations of the {\em beeping model} of communication. In each round an agent can either listen, i.e., stay silent,
or beep, i.e., emit a signal. In the {\em local beeping model},
an agent hears a beep in a round if it listens in this round and if the other agent is at the same node and beeps. In the {\em global beeping model},
an agent hears a {\em loud} beep in a round if it listens in this round and if the other agent is at the same node and beeps, and it hears a {\em soft}
beep in a round if it listens in this round and if the other agent is at some other node and beeps.

 The beeping model has been introduced
in \cite{CK} for vertex coloring, and used
in \cite{AABCHK} to solve the MIS problem, and in  \cite{YJYLC} to construct a minimum connected dominating set.
In \cite{GN}, the authors studied the quantity of computational resources needed to solve problems in complete networks using beeps. In the variant from  the above papers the beeping entities were nodes rather than agents, and beeps of a node were
heard at adjacent nodes. The beeping model is widely applicable, as it makes small demands on communicating devices, relying only on carrier sensing.
In the case of mobile agents, both the local and the global beeping models are applicable in different settings. The local model is applicable even for agents
having very weak transmissions capabilities, limiting reception of a beep to the same node. The global model is applicable for more powerful agents,
that can beep sufficiently strongly to be heard in the entire network, and having a listening capability of differentiating a beep emitted at the same node from a beep emitted at a different node.   

It should be noted that our local beeping model  is arguably the weakest way of communication between agents: they can communicate
only when residing simultaneously at the same node, they cannot hear when they beep, and messages are the simplest possible. In fact, as mentioned in 
\cite{CK}, local beeps are an even weaker way of communicating than using one-bit messages, as the latter ones allow three different states (0,1 and no message),
while local beeps permit to differentiate only between a signal and its absence. 
Clearly, without any communication, rendezvous with detection is impossible, as agents cannot become aware
of each other's presence at a node. Notice also that in the global beeping model it would not be possible to remove the distinction between hearing a 
loud beep when the beeping agent is at the same node and hearing a soft beep when the beeping agent is at a different node. Indeed, the same strength of 
beep reception would make it impossible for an agent $A_1$ to inform the other agent $A_2$ of the presence of $A_1$ at the same node, and hence rendezvous with detection would be impossible. The global beeping model is at least as strong as the local one, in the sense that any algorithm of rendezvous with detection working in the local model works also in the global model, by simply ignoring soft beeps. We will see that the converse is not true.   
%Let us also mention that the necessity of communicating the presence of an agent at a node may be vital
%e.g., in applications where agents are mobile robots navigating in a dark labyrinth, 
%whose corridors are edges of the network and nodes are their crossings.

For a given network,
the execution time of an algorithm of rendezvous with detection, for agents with given labels starting in given rounds from given initial positions,  is the number of rounds from the activation of the later agent to the declaration of rendezvous. For a given class of networks, the {\em time} of an algorithm of rendezvous with detection is its worst-case execution time, over all networks in the class, all initial positions, all pairs of distinct labels and all starting times.

%--------------------------------------------------
{\bf Our results.}
Our first result  answers the basic question: Is it possible to achieve rendezvous with detection in arbitrary networks, and if so, how fast it can be done?
We present a deterministic algorithm of rendezvous with detection working, even for the local beeping model, in an arbitrary unknown network in time polynomial in the size of the network
and in the length of the smaller label (i.e., in the logarithm of this label). In fact, the time complexity of our algorithm matches that of the fastest, known to date, rendezvous
algorithm without detection, constructed in \cite{TSZ07}.

However, in this algorithm, agents spend a lot of energy: the number of moves
that an agent must make, is proportional to the time of rendezvous. 
On the other hand, in many applications, e.g., when agents are mobile robots, they are battery-powered devices, and hence the energy that an agent can spend on moves is limited. 
It is thus natural to ask if {\em bounded-energy agents}, i.e., agents that can make at most 
$c$ moves, for some integer $c$, can always achieve rendezvous with detection as well. This is impossible for some networks of unbounded size. Hence we rephrase the question: Can bounded-energy agents always achieve rendezvous with detection in bounded-size networks? We prove that the answer to this question is positive, even in the local beeping model but, perhaps surprisingly, this ability comes at a steep price of time: the meeting time of bounded-energy agents is {\em exponentially} larger
than that of unrestricted agents.
By contrast, we show an algorithm for rendezvous with detection in the global beeping model that works for bounded-energy agents
(in bounded-size networks) as fast as for unrestricted agents.

%--------------------------------------------------
{\bf Related work.}
The vast literature on rendezvous can be divided according to the mode in which agents move (deterministic or randomized) and
the environment where they move (a network modeled as a graph or a terrain in the plane).
An extensive survey of  randomized rendezvous in various scenarios  can be found in
\cite{alpern02b}, cf. also  \cite{anderson98b,KKPM08}. 
%Several authors
%considered the geometric scenario (rendezvous in an interval of the real line, see, e.g.,  \cite{baston01,gal99},
Rendezvous of two or more agents in the plane has been considered e.g.,  in \cite{fpsw,FSVY}. 

Our paper is concerned with deterministic rendezvous in networks, surveyed in \cite{P}.
In this setting a lot of effort has been dedicated to the study of the feasibility of rendezvous, and to the time required to achieve this task, when feasible. For instance, deterministic rendezvous with agents equipped with tokens used to mark nodes was considered, e.g., in~\cite{KKSS}. Time of deterministic rendezvous of agents equipped with unique labels was discussed in \cite{DFKP,TSZ07}. Memory required by the agents to achieve deterministic rendezvous has been studied in \cite{BIOKM,FP2} for trees and in  \cite{CKP} for general graphs. In \cite{MP} the authors studied tradeoffs between the time of rendezvous and the total number
of edge traversals by both agents until the meeting. 

Apart from the synchronous model used in this paper, several authors have investigated asynchronous rendezvous in the plane \cite{CFPS,fpsw,FSVY} and in network environments
\cite{BCGIL,DPV}.
In the latter scenario the agent chooses the edge which it decides to traverse but the adversary controls the speed of the agent. Under this assumption rendezvous
in a node cannot be guaranteed even in very simple graphs, and hence the rendezvous requirement is relaxed to permit the agents to meet inside an edge.

%%%%%%%%%%%%%%%%%%%%%%%%%%%%%%%%%%%%%%%%%%%%%%%%%%%%%%%%%%%
\section{Preliminaries}
%%%%%%%%%%%%%%%%%%%%%%%%%%%%%%%%%%%%%%%%%%%%%%%%%%%%%%%%%%%
In the rest of the paper the word ``graph'' means a simple connected undirected graph
modeling a network. The {\em size} of a graph is the number of its nodes.
In this section we recall two procedures known from the literature, that will be used as building blocks in our algorithms. 
The aim of the first procedure is graph exploration, i.e., visiting all nodes of a graph by a single agent. 
The procedure, called $EXP(m)$, is based on universal exploration sequences (UXS) \cite{Ko}, and follows from the  result of Reingold \cite{Re}. Given any positive integer $m$, it allows the agent to visit all nodes of any graph of size at most $m$,
starting from any node of this graph, using $R(m)$ edge traversals, where $R$ is some polynomial. 

A UXS is an infinite sequence $x_1,x_2,\dots$ of non-negative integers. Given this sequence, whose effective construction follows from  \cite{Ko,Re}, the procedure $EXP(m)$ can be described as follows.
In step 1, the agent leaves the starting node by port 0. For $i \geq 1$, the agent that entered the current node 
of degree $d$ by some port $p$ in step $i$,
computes the port $q$ by which it has to exit in step $i+1$ as follows: $q=(p+x_i)\mod d$.
The result of Reingold implies that if an agent starts at any node $v$ of an arbitrary graph with at most $m$ nodes,  and applies procedure $EXP(m)$, then it will visit all nodes of the graph after $R(m)$ steps.

The second procedure, due to Ta-Shma and Zwick
\cite{TSZ07},  guarantees rendezvous (without detection) in an arbitrary graph. 
Below we briefly sketch this procedure, which will
be used in our algorithm of rendezvous with detection for unrestricted agents.

Let $\mathbb{Z}^+$ denote the set of positive integers and let $\mathbb{Z}^*$ denote the set of integers greater or equal than $-1$.
For any positive integer $L$, Ta-Shma and Zwick define a function $\Phi_L: \mathbb{Z}^+ \times \mathbb{Z}^+ \times \mathbb{Z}^* \longrightarrow \mathbb{Z}^*$.
Intuitively, this function describes a walk in the graph: when an agent starts at some node $v$ of the graph, the function $\Phi$ indicates which port the agent should take in the $t$-th step of its walk, or that it should stay idle (i.e., do not move) in the $t$th step. The three parameters of the function are: the next step number $t$, the degree $d$ of the current node, and the port number $p$ by which the agent entered the current node in the previous step, or $-1$ if it stayed idle in the previous step. The value of the function is either the port number by which the agent should leave the current node in the next step, or $-1$ if it should stay in the next step.

More formally, the function $\Phi_L$ is {\em applied} by an agent with label $L$ in a graph $G$ at a node $v$ of $G$ as follows. Let $v_0=v$ and let $v_1$ be the node adjacent to $v_0$, such that the edge $\{v_0,v_1\}$ has port number 0 at $v_0$. Suppose that nodes $v_0,v_1,\dots, v_{t-1}$ are already constructed, so that $v_{i+1}$ either
equals $v_i$ or is adjacent to $v_i$. The node $v_t$ is defined as follows. In the case when $v_{t-1}=v_{t-2}$ and the degree of $v_{t-1}$ is $d$, then $v_t=v_{t-1}$ if
$\Phi_L(t,d,-1)=-1$; if $\Phi_L(t,d,-1)=q\geq 0$ then $v_t$ is the node adjacent to $v_{t-1}$ such that the port number at $v_{t-1}$ corresponding to edge $\{v_{t-1},v_t\}$ is $q$.
In the case when $v_{t-1}\neq v_{t-2}$, the port number at $v_{t-1}$ corresponding to edge $\{v_{t-1},v_{t-2}\}$ is $p$  and the degree of $v_{t-1}$ is $d$, 
then $v_t=v_{t-1}$ if
$\Phi_L(t,d,p)=-1$; if $\Phi_L(t,d,p)=q>0$ then $v_t$ is the node adjacent to $v_{t-1}$ such that the port number at $v_{t-1}$ corresponding to edge $\{v_{t-1},v_t\}$ is $q$.
Hence the application of function $\Phi_L$ at node $v$ defines an infinite walk of the agent with label $L$ in the graph $G$. This walk starts at $v$ and in each round $t$
the agent either stays at the current node or moves to an adjacent node by a port determined by the function $\Phi_L$ on the basis of the degree of the current node
and of the port by which the agent entered it. A round $t$ is called {\em active} for the agent if $v_t \neq v_{t-1}$ and it is called {\em passive} if $v_t=v_{t-1}$. 

The following result, proved in  \cite{TSZ07},  guarantees rendezvous without detection in polynomial time, if two agents apply functions $\Phi_L$ corresponding to their labels, in an unknown graph.

\begin{theorem}\label{tz}
There exists a polynomial $P$ in two variables, with the following property.
Let $G$ be an $n$-node graph and consider two agents with distinct labels $L_1$, $L_2$ respectively, 
starting at nodes $v$ and $w$ of the graph in rounds $t_1 \geq t_2$. Let $t\geq t_1$ and
let $\ell$ be  the smaller label. If agent with label $L_i$ applies function $\Phi_{L_i}$ at its starting node, for $i=1,2$, then
agents are simultaneously at the same node in some round of the time interval $[t,t+P(n,\log \ell)]$. Moreover, rendezvous occurs in a round which is active for one of the agents and passive for the other. The same property remains true if one of the agents stays idle and the other agent applies its function~$\Phi_{L_i}$.
\end{theorem}

\section{Rendezvous with detection of unrestricted agents}

In this section we describe and analyze an algorithm of rendezvous with detection which works for unrestricted agents, i.e., for agents that can spend an arbitrary amount
of energy on moves. It works even for the weaker of our two models, i.e., for the local beeping model.
Our algorithm uses the following procedure which describes an infinite walk of an agent with label $L$, 
based on the above described application of the function $\Phi_L$. 

{\bf Procedure} {\tt Beeping walk}

Consider an agent with label $L$ starting at node $v$ of a graph $G$. Let $W$ be the walk resulting from the application of  $\Phi_L$ in graph $G$ at node $v$.
Each round of $W$ is replaced by 2 consecutive rounds as follows. If round $t$ of $W$ is passive, i.e., $v_t=v_{t-1}$, then this round is replaced by two rounds in which
the agent stays at $v_t$ and listens.  If round $t$ of $W$ is active, i.e., $v_t\neq v_{t-1}$, then this round is replaced by the following two rounds: in the first of these
rounds the agent goes to $v_t$ and beeps, and in the second of these rounds the agent stays at $v_t$ and listens.

We now describe our algorithm for rendezvous with detection. It is executed by each agent. Note that the execution of procedure  {\tt Beeping walk}, called by the algorithm,
depends on the label of the agent.

\begin{center}
\fbox{
\begin{minipage}{12cm}
{\bf Algorithm} {\tt RV-with-detection}

Perform procedure {\tt Beeping walk} {\bf until} you hear a beep\\
Let $s$ be the current round number 
(counted since your wake-up)\\
Stay idle forever\\
Beep in round $s+1$
listen in round $s+2$\\
{\bf If} you hear no beep in round $s+2$ {\bf then}\\ 
\hspace*{1cm}declare rendezvous in round $s+3$ and stop\\
{\bf else}\\
\hspace*{1cm}listen in round $s+3$, declare rendezvous in round $s+4$, and stop.   
\end{minipage}
}
\end{center}

We now show that Algorithm {\tt RV-with-detection} correctly accomplishes rendezvous with detection and works in time polynomial in the size of the graph and in the
logarithm of the smaller label. The agent that starts later will be called the {\em later} agent and the other one the {\em earlier} agent.
If agents start simultaneously, these qualifiers are attributed arbitrarily.

\begin{theorem}\label{ub1}
Consider two agents with distinct labels $L_1$, $L_2$ respectively, 
starting at nodes $v$ and $w$ of an $n$-node graph in possibly different rounds.
Let $\ell$ be  the smaller label.
If both agents execute Algorithm {\tt RV-with-detection}, then they meet and simultaneously declare rendezvous in time $O(P(n,\log \ell))$, i.e., polynomial in $n$ and in $\log \ell$, after
the start of the later agent.
\end{theorem}

\begin{proof}
An agent executing Algorithm {\tt RV-with-detection} starts by performing procedure {\tt Beeping walk}. We first prove that when both agents execute this procedure
there must be a round in which one of them hears the beep of the other. By Theorem \ref{tz},  when agent with label $L_i$ applies function $\Phi_{L_i}$ at its starting node, for $i=1,2$, there is a round $r$ when they meet, and this round is active for one agent and passive for the other. 
(The names of the rounds are for the ease of description only, as none of the agents knows them.) 
This holds regardless of the starting rounds and starting positions  of the agents.
 By definition, 
procedure {\tt Beeping walk} is a simulation of the application of $\Phi_{L_i}$ with 2-round segments corresponding to rounds. A segment simulating an active (resp. passive) round will be called an {\em active (resp. passive) segment}. 

Suppose that the segment simulating the meeting round $r$ is the $\rho$th segment of the earlier agent. 
If the segments of the two agents are aligned, i.e., the first round of the later agent is the first round of a segment of the earlier agent, then in the first round of the $\rho$th
segment of the earlier agent one agent beeps and the other listens because this segment is active for one of them and passive for the other.
Since both agents are at the same node in this round, one of them hears the beep of the other. Hence we may assume
that the segments are not aligned. Suppose that the $\sigma$th segment of the later agent starts during the $\rho$th segment of the earlier agent, 
i.e., the first round of the  $\sigma$th segment of the later agent is the second round of the $\rho$th segment of the earlier agent. 
By Theorem \ref{tz} there are two cases: either the $\rho$th segment of the earlier agent is passive and the $\sigma$th segment of the later agent is active, or
the $\rho$th segment of the earlier agent is active and the $\sigma$th segment of the later agent is passive. In the first case let $r'$ be the first round of the $\sigma$th
segment of the later agent. In round $r'$ both agents are at the same node and the later agent beeps while the earlier agent listens, hence it hears the beep. Consider the
second case. Regardless of whether the $(\sigma-1)$th segment of the later agent is active or passive, during this segment the later agent is at the same node as the earlier
agent during its $\rho$th segment. Let $r''$ be the first round of the $\rho$th segment of the earlier agent.  Since segments of agents are not aligned, in round $r''$, when the earlier agent beeps, the later agent listens, and hence hears the beep.
This shows that in all cases there must be a round in which one of the agents hears the beep of the other.

Let $t$ be the first round in the execution of Algorithm {\tt RV-with-detection} by both agents, in which one agent hears the beep of the other. 
Denote by $A_1$ the agent that listens in round $t$ and by $A_2$ the other agent. By the algorithm, agent $A_1$ stays idle from round $t$ on, and beeps 
in round $t+1$. In this round agent $A_2$  stays idle at the same node and listens, as it still executes procedure {\tt Beeping walk}, and this is the second round
of a segment for this agent. Hence agent $A_2$ hears a beep in round $t+1$. This is the first beep that it hears. Consequently, agent $A_2$ stays 
idle from round $t+1$ on, and beeps in round $t+2$. Hence agent $A_1$ hears a beep in round $t+2$. By the algorithm, it listens in round $t+3$, declares rendezvous in
round $t+4$ and stops. As for agent $A_2$, it listens in round $t+3$ and does not hear a beep in this round. Hence it also declares rendezvous in round $t+4$ and stops.

It follows that in round  $t+4$ both agents are at the same node, declare rendezvous and stop. This proves the correctness of
Algorithm {\tt RV-with-detection}. It remains to estimate its time. Round $t$ must occur at most $2P(n, \log \ell)$ rounds after the start
of the later agent. 
Since simultaneous declaration of rendezvous is in round $t+4$,
this proves the theorem. 
\end{proof}

\section{Rendezvous with detection of bounded-energy agents}

In this section we study rendezvous with detection of agents that can perform a bounded number of moves. Let $c$ be a positive integer.
A $c$-{\em bounded agent} is defined as an agent that can perform at most $c$ moves. (Notice that we do not restrict the number of beeps; indeed, the amount
of energy required to make a move is usually so much larger than the amount of energy required to beep that ignoring the latter seems to be a reasonable approximation of reality in many applications.) 
Can $c$-bounded agents, for some integer $c$, perform rendezvous with detection in arbitrary graphs? The answer to this question is, of course, negative, even if detection is not required. For any integer $c$,  $c$-bounded agents starting at  distance
larger than $2c$ cannot meet because at least one of them would have to make more than $c$ steps. Even if we assume that the initial distance between the agents is 1,
meeting of $c$-bounded agents is impossible in some graphs. Indeed, consider two $n$-node stars whose centers are linked by an edge, with agents starting at the centers of the stars. In the worst case, at least one of the agents must make at least $n-1$ steps before meeting (to find the connecting edge), which is impossible for $c$-bounded agents, when $n$ is large.

Thus, we rephrase the question: Can $c$-bounded agents always achieve rendezvous with detection in bounded-size graphs? More precisely, for any integer $n$,
does there exist an integer $c$, such that $c$-bounded agents can achieve rendezvous with detection in all graphs of size at most $n$? (Notice that, for example,
Algorithm {\tt RV-with-detection} cannot be used here. In this algorithm, the number of steps  performed by an agent with label $L$ is proportional to $P(n,\log L)$, and hence, even when the size of the graph is bounded, this number can be arbitrarily large.)   
The answer to our question turns out to be positive, even in the local beeping model. Below we describe an algorithm that performs this task. 

\subsection{Bounded-energy agents in the local beeping model}

Our algorithm uses the following procedure, for an integer parameter $n$.

{\bf Procedure} {\tt Beeping exploration} $(n)$

Let $EXP(n)$ be the procedure described in Section 2 that permits exploration of all graphs of size at most $n$. Replace each round $r$ of $EXP(n)$ by three consecutive rounds as follows. If in round $r$ of $EXP(n)$ the agent takes port $p$ to move to node $w$, then in the first of the three replacing rounds the agent
takes port $p$ to move to $w$ and beeps, and in the second and third of the replacing rounds it stays at $w$ and listens. 

Hence, in each of the three rounds replacing a round $r$ of $EXP(n)$, the agent is at the same node in Procedure  {\tt Beeping exploration} $(n)$ as it is in Procedure
$EXP(n)$  in round $r$.

We now describe our algorithm for rendezvous with detection of bounded-energy agents, executed by an agent with label $L$ in a graph of size at most $n$. Recall that $R(n)$ is the execution time of $EXP(n)$.
The idea of the algorithm is the following. Its main part, called {\em  block}, consists of two executions of Procedure  {\tt Beeping exploration} $(n)$ between which a long {\em waiting period} 
is inserted, during which the agent is silent (it listens) and stays idle. The length of this period depends on the label of the agent. We will prove that, regardless of the delay
between the starting times of the agents, an entire execution of Procedure  {\tt Beeping exploration} $(n)$ of one of the agents must either fall within the waiting period of the other agent,
or must be executed after both executions of this procedure by the other agent.
The block of the algorithm executed by a given agent is interrupted in one of the two cases: either when (a) the agent hears a beep during its waiting period or after completing its  block, or when (b) it hears beeps in two consecutive rounds during one of the executions of Procedure  {\tt Beeping exploration} $(n)$. In case (a) the agent responds by beeps in two consecutive rounds,
declares rendezvous in the next round and stops. In case (b) it declares rendezvous in the next round and stops.

Below we give the pseudo-code of the algorithm executed by an agent with label $L$ in a graph of size at most $n$.
During the executions of Procedure  {\tt Beeping exploration} $(n)$, a boolean variable $waiting$ is set to false, and
during the waiting period and after the second execution of Procedure  {\tt Beeping exploration} $(n)$  it is set to true. We use a boolean valued function {\tt condition}
which takes the variable $waiting$ as input, and returns, after each round, the boolean value of the expression\\
($waiting$ {\bf and} you hear a beep) {\bf or}  ($\neg waiting$ {\bf and} you hear beeps in two consecutive rounds). \\

\begin{center}
\fbox{
\begin{minipage}{16cm}

{\bf Algorithm} {\tt Bounded-energy-RV-with-detection}\\

$waiting :=$ false\\
Perform the following sequence of actions in consecutive rounds\\ and verify the value of {\tt condition} in each round\\ 
{\bf until}  {\tt condition} becomes true \\
\hspace*{1cm}Perform Procedure  {\tt Beeping exploration} $(n)$\\
\hspace*{1cm}$waiting :=$ true\\
\hspace*{1cm}Stay idle and listen for $6L \cdot R(n)$ rounds\\
\hspace*{1cm}$waiting :=$ false\\
\hspace*{1cm}Perform Procedure  {\tt Beeping exploration} $(n)$\\
\hspace*{1cm}$waiting :=$ true\\
\hspace*{1cm}Stay idle forever and listen\\
$s:=$  the round number when {\tt condition} becomes true (counted since your wake-up)\\
{\bf if} $waiting$ {\bf then}\\
\hspace*{1cm}beep in rounds $s+1$ and $s+2$\\
\hspace*{1cm}declare rendezvous in round $s+3$ and stop\\
{\bf else}\\
\hspace*{1cm}declare rendezvous in round $s+1$ and stop.
\end{minipage}
}
\end{center}

\begin{theorem}\label{ub2}
For any positive integer constant $n$ there exists a positive integer $c$, such that Algorithm {\tt Bounded-energy-RV-with-detection} can be executed by $c$-bounded agents in any graph of size at most $n$. If two such agents with distinct labels execute this algorithm in such a graph, then they meet and simultaneously declare rendezvous in time $O(\ell ^*)$
after the start of the later agent, where $\ell ^*$ is the larger label.  
\end{theorem}

\begin{proof}
Let $n$ be a fixed positive integer. In order to show the existence of the integer $c$, it is enough to show that the number of steps prescribed by the algorithm depends only on
the integer $n$ (and not on the label $L$ of the agent). The agent moves  only during the two executions 
of Procedure  {\tt Beeping exploration} $(n)$, at most once every three rounds. Since a full execution of 
Procedure  {\tt Beeping exploration} $(n)$ takes $E=3R(n)$ rounds, the agent moves at most $2E/3=2R(n)$ times, and hence it is enough to take $c=2R(n)$.

We now prove the correctness of the algorithm. Let $B(L)$ be the sequence of actions

\hspace*{1cm}Perform Procedure  {\tt Beeping exploration} $(n)$;\\
\hspace*{1cm}Stay idle for $6L \cdot R(n)$ rounds and listen; (waiting period)\\
\hspace*{1cm}Perform Procedure  {\tt Beeping exploration} $(n)$;\\

For any label $L$, the sequence $B(L)$ of actions will be called the {\em  block} of the agent with label $L$.
We will use the following claim.

{\bf Claim.}
Consider blocks $B(L_1)$ and $B(L_2)$, for $L_1>L_2$, arbitrarily shifted in time with respect to each other. Then one of the two following properties must hold.
Either an entire execution of Procedure  {\tt Beeping exploration} $(n)$ in one of the blocks falls within the waiting period of the other block, or
an entire execution of Procedure  {\tt Beeping exploration} $(n)$ in one of the blocks falls after the end of the other block.

In order to prove the claim, consider  blocks $B(L_1)$ and $B(L_2)$, for $L_1>L_2$. Call the block $B(L_1)$ the larger block and the block $B(L_2)$ the smaller block. The waiting period of the smaller block has length $Y=2EL_2 \geq 2E$. Since $L_1 \geq L_2+1$,  the waiting period of the larger block has length 
$2EL_1 \geq 2E(L_2+1)=2EL_2+2E= Y+2E$. Let the starting round of the block that starts earlier have number 0. 
We will use global round numbers starting from this round.  (This is for the purpose of analysis only, as agents do not have access to any global round counter.)

If the smaller block starts earlier, then the second execution of Procedure  {\tt Beeping exploration} $(n)$ of the larger block is performed after the end of the smaller block, and the claim is proved. Hence we may assume that the larger block starts earlier or simultaneously with the smaller block.
Let $p$ be the first round of the smaller block. Consider three cases.

Case 1. $0 \leq p \leq E$.\\
In this case the second execution of Procedure  {\tt Beeping exploration} $(n)$ in the smaller block falls entirely within the waiting period of the larger block.

Case 2. $E<p \leq 2L_1E$.\\
In this case the first execution of Procedure  {\tt Beeping exploration} $(n)$ in the smaller block falls entirely within the waiting period of the larger block.

Case 3. $p>2L_1E$.\\
In this case the second execution of Procedure  {\tt Beeping exploration} $(n)$ in the smaller block falls entirely after the end of the larger block.

This proves the claim.

Consider agents with labels $L_1>L_2$. Call the agent with label $L_1$ the larger agent and the agent with label $L_2$ the smaller agent.
The claim implies that there must exist a round $r$ in which {\tt condition} is satisfied (i.e., has value true) for one of the agents. Indeed,  if such a round has not occurred previously, it must occur
during the execution of Procedure  {\tt Beeping exploration} $(n)$ by one of the agents, call it $A_1$, that falls entirely either in the waiting period of the other agent, call it $A_2$, or after the 
end of the second execution of Procedure  {\tt Beeping exploration} $(n)$ by agent $A_2$. In both situations, agent $A_2$ has its variable $waiting$ set to true, it is idle and listens during an entire execution of Procedure  {\tt Beeping exploration} $(n)$ by agent $A_1$. During this execution, agent $A_1$ visits all nodes of the graph and beeps at each node. Hence agent $A_2$ hears a beep while its variable $waiting$ is set to true, which means that {\tt condition} is true for agent $A_2$. 

Let $r_0$ be the first round in which {\tt condition} is true for some agent. We will show that for this agent the first possibility, i.e., ($waiting$ {\bf and} you hear a beep)
must be satisfied in round $r_0$. Indeed, suppose that the second possibility, i.e., ($\neg waiting$ {\bf and} you hear beeps in two consecutive rounds) is satisfied in round $r_0$. These beeps in two consecutive rounds that the agent heard must have been produced in rounds $r_0-1$ and $r_0$ by the other agent. It could not be during the execution of its block because then the agent beeps only during the execution of Procedure  {\tt Beeping exploration} $(n)$ and this never happens in two consecutive rounds. 
Hence it must have happened after the block of the other agent has been interrupted. This is possible only when {\tt condition} has been true for the other agent, which
must have occurred before round $r_0$, contradicting the definition of this round.

Consider round $r_0$ and the agent, call it $A_3$, for which the clause ($waiting$ {\bf and} you hear a beep) is satisfied in round $r_0$. Agent $A_3$ is at some node $v$ in round $r_0$.
In round $r_0$, the other agent, call it $A_4$, must
still execute its  block, and more precisely it must execute Procedure  {\tt Beeping exploration} $(n)$. It is also at node $v$ in round $r_0$. 
Hence it remains idle at $v$ and listens in the two rounds after it has beeped,
i.e., in rounds $r_0+1$ and $r_0+2$. In these rounds agent $A_3$ stays idle at $v$ and beeps. Then it declares rendezvous in round $r_0+3$ (still remaining at node $v$) and stops.
As for agent $A_4$, after hearing beeps in rounds $r_0+1$ and $r_0+2$, the clause ($\neg waiting$ {\bf and} you hear beeps in two consecutive rounds) becomes
satisfied for it in round $s=r_0+2$. Hence agent $A_4$ (still remaining at node $v$)  declares rendezvous in round $s+1=r_0+3$ and stops. This concludes the proof of correctness.

It remains to estimate the time between the start of the later starting agent and the declaration of rendezvous. Since $n$ is constant, $R(n)$ is also constant, and hence the duration of the execution of its  block by an agent with label $L$ is $(2L+2)\cdot 3R(n) \in O(L)$. Let $L_i$, where $i=1$, or $i=2$, be the label of the later starting agent.
Hence the time 
between the start of this agent and the declaration of rendezvous is at most $(2L_i+2)\cdot 3R(n)+3 \in O(L_i)=O(\ell ^*)$ because $L_i\leq \ell ^*$. This concludes the proof.
\end{proof}

It is interesting to compare the time sufficient to complete the task of rendezvous with detection, given by Algorithm {\tt RV-with-detection} for unrestricted agents, with the time
given by Algorithm {\tt Bounded-energy-RV-with-detection} for bounded-energy agents. This comparison is meaningful on the class of graphs for which both types of agents can achieve rendezvous with detection, i.e., for graphs of bounded size. Consider the class $C_n$ of graphs of size at most $n$, for some constant $n$, and consider $c$-bounded agents for some integer $c$ large enough to achieve rendezvous with detection on the class $C_n$ using Algorithm {\tt Bounded-energy-RV-with-detection}.
By Theorem \ref{ub1}, unrestricted agents can accomplish rendezvous with detection in time $O(P(n,\log \ell))$, i.e., since $n$ is constant, in time {\em polylogarithmic in the smaller label}. By contrast, by Theorem \ref{ub2}, bounded-energy agents can accomplish rendezvous with detection in time $O(\ell ^*)$, i.e., {\em linear in the larger label}.
It is natural to ask if this  exponential gap in time, due to energy restriction, is unavoidable. The following lower bound shows that the answer to this question is yes. In fact, this lower bound holds even 
for the two-node graph, even with simultaneous start of the agents, and even for rendezvous without detection. 

\begin{theorem}\label{lb}
Let $c$ be a positive constant. In the local beeping model,
the time of rendezvous on the two-node graph of $c$-bounded agents with labels from the set $\{1,\dots , M\}$ is $\Omega (\sqrt[c]{M})$. 
\end{theorem}

\begin{proof}
Consider any rendezvous algorithm $\cal A$ of $c$-bounded agents on the two-node graph. Suppose that labels of the agents can be drawn from the set $\{1,\dots , M\}$
and that agents start simultaneously. Let $T$ be the worst-case meeting time over all pairs of labels from the set $\{1,\dots , M\}$. For any label $L\in \{1,\dots , M\}$,
let $\Phi _L : \{1,\dots , T\} \longrightarrow \{0,1\}$ denote the binary sequence of length $T$ with the following meaning. If agent with label $L$ executes alone algorithm $\cal A$
in the two-node graph (we call it the {\em solo execution}), then it moves in round $i$, if $\Phi _L(i)=1$, and it stays idle in round $i$, if $\Phi _L(i)=0$. For any label $L$, the function $\Phi _L $ is well-defined
because in the two-node graph the history of the agent in any round is exactly the binary sequence describing the previous moves (i.e., the agent cannot ``learn'' anything from the environment during the navigation in this graph, as opposed to more complicated graphs in which it could learn degrees of visited nodes or port numbers by which it enters them). Hence a solo execution of algorithm $\cal A$ by an agent in the two node graph depends only on the label of the agent.
Notice that, if there are two agents in this graph, executing algorithm $\cal A$, then the behavior of each of them before the meeting is the same as in the solo execution.

%For any label $L\in \{1,\dots , M\}$, let $\alpha (L)$ denote the minimum integer $r$ from the set $\{1,\dots , T\}$, such that the agent with label $L$ meets an agent
%with any label $L'\in \{1,\dots , M\}$ different from $L$, by round $r$. There exists a subset ${\cal L} \subseteq  \{1,\dots , M\}$ of size at least $M'=M/T$, such that for any %label $L \in \cal L$, the integer $\alpha (L)$ is the same. Call this integer $S$. For any $L \in \cal L$, let $\Psi_L$ be the restriction of the function $\Phi_L$ to the set
%$\{1,\dots , S\}$. Functions $\Psi_L$ describe the solo execution of agents with labels from the set $\cal L$ until round $S$. By definition of $S$, for each 
 %label $L \in \cal L$ there exists a label $L'\neq L$ in the set $\{1,\dots , M\}$, such that agents with labels $L$ and $L'$ do not meet before round $S$.
 Since agents are
$c$-bounded, the number of values 1 in each function $\Phi _L $ is at most $c$.
The number of such functions is 
$${{T}\choose {0}} + {{T}\choose {1}}+\cdots + {{T}\choose {c}}\leq 1+T+T^2 + \cdots + T^c\leq 2T^{c}.$$
If $2T^c <M$, there would exist two labels $L_1,L_2 \in \{1,\dots , M\}$, for which $\Phi _{L_1}=\Phi _{L_2}$. Agents with these labels could not meet by round $T$, as they would
move exactly in the same rounds until round $T$, hence in every round they would be at different nodes. It follows that $2T^c \geq M$, hence $T \in \Omega (\sqrt[c]{M})$.
\end{proof}

Theorem \ref{lb} implies that in the local beeping model, any rendezvous algorithm for bounded-energy agents must have time at least $\Omega (\sqrt[c]{\ell ^*})$, where $\ell ^*$ is the larger label and $c$ is some constant.
Theorems  \ref{ub1},  \ref{ub2} and \ref{lb} imply the following corollary.

\begin{corollary}
Rendezvous with detection of bounded-energy agents is feasible in the class of bounded-size graphs in the local beeping model, but its time must be exponentially larger than the best time
of rendezvous with detection of unrestricted agents in this class of graphs.
\end{corollary}

\subsection{Bounded-energy agents in the global beeping model}

Our final result shows that in the stronger of our two models, i.e., the global beeping model, the lower bound on time proved in Theorem  \ref{lb} does not
hold anymore. In fact, we show that in this model, bounded-energy agents can meet with detection in the class of bounded-size graphs in time logarithmic in the smaller label. We will also prove that this time is optimal even in the two-node graph. 

The high-level idea of the algorithm is to first break symmetry between the agents in time logarithmic in the smaller label, without making any moves, using the possibility of hearing the beeps of the other agent, wherever it is in the graph. This is the purpose of Procedure  {\tt Symmetry-breaking} in which agents beep or listen according to bits of their transformed label (the transformation will be explained below). Moreover, at the end of this procedure, both agents declare the same round to be {\em red} which permits  them to synchronize in the next part of the algorithm.

 After breaking symmetry, one of the agents remains idle and the other agent finds it using a bounded number of moves. This is done during Procedure
  {\tt Modified-beeping-exploration}, when the moving agent performs exploration of the graph while beeping in every second round. This procedure is started by the moving agent in round {\em red}, simultaneously declared by both agents in Procedure  {\tt Symmetry-breaking}.
Correct declaration of rendezvous is possible due to the distinction between  hearing {\em loud} and {\em soft} beeps, and due to synchronization of agents which permits each of them to hear the beeps of the other agent.

We now proceed with the detailed description of the algorithm.
We first define the following transformations of the label $L$ of an agent. Let $(c_1\dots c_k)$ be the binary representation of the label $L$.
Let $T_1(L)$ be the binary sequence $(01c_1c_1c_2c_2\dots c_kc_k01)$, and let $T_2(L)$ be the result of replacing each bit 0 of $T_1(L)$ by the string $(00)$ and
each bit 1 by the string $(10)$. Note that the length of the binary string $T_2(L)$ is $2(2k+4)\in O(\log L)$.

The following procedure, executed by an agent with label $L$ and called upon the activation of the agent, does not involve any moves and permits to break symmetry between any two agents
with different labels.

\begin{center}
\fbox{
\begin{minipage}{13cm}

{\bf Procedure} {\tt Symmetry-breaking}\\

Let $T_2(L)=(d_1\dots d_s)$\\
$i:=1$\\
{\bf repeat} in consecutive rounds {\bf until} you hear a beep\\
\hspace*{1cm}{\bf if} ($i \leq s$ {\bf and} $d_i=1$) {\bf then} beep {\bf else} listen\\
\hspace*{1cm}$i:=i+1$\\
Let $r$ be the round when you first hear a beep  (counted since your wake-up)\\
{\bf if} you beeped in round $r-1$ {\bf then}\\
\hspace*{1cm}declare round $r+1$ {\em red}\\
\hspace*{1cm}$role:=waiting$\\
 {\bf else}\\ 
\hspace*{1cm}beep in round $r+1$\\
\hspace*{1cm}declare round $r+2$  {\em red}\\
\hspace*{1cm}$role:=walking$;\\
{\bf if} the beep you heard was {\em loud} {\bf then}\\ 
\hspace*{1cm}declare rendezvous in the {\em red} round and stop;

\end{minipage}
}
\end{center}

\begin{lemma}\label{breaking}
Upon completion of Procedure {\tt Symmetry-breaking}, both agents declare the same round to be {\em red}.
For one of the agents round {\em red} is the next round after it heard a beep for the first time, and this agent sets $role:=waiting$. 
For the other agent round {\em red} is two rounds
after it heard a beep for the first time, and this agent sets $role:=walking$.  The round declared {\em red} is $O(\log \ell)$ rounds after the activation of the later agent, where $\ell$ is the smaller label.
\end{lemma}

\begin{proof}
We first prove that there exists a round $\xi$ in which both agents are present in the graph and one of the agents hears a beep.
Let $L_1$ be the label of agent $A_1$ and let $L_2$ be the label of agent $A_2$.
Consider three cases. In all of them we show a round $\xi$ as above, assuming that none of the agents heard a beep at any earlier round.

{\bf Case 1.} Both agents are activated in the same round and have labels with binary representations of  equal length.

Let $i$ be the first index for which strings $T_2(L_1)$ and  $T_2(L_2)$ differ. Since agents are activated in the same round, the same round corresponds
to index $i$ for both of them. In this round one of the agents beeps and the other agent listens, hence it hears a beep.

{\bf Case 2.} Both agents are activated in the same round and have labels with binary representations of  different length.

Without loss of generality, let $L_1$ be the label of smaller length. Let $T_2(L_1)=(d_1\dots d_s)$. Let $\alpha$ be the common activation round of the agents
and let $\gamma=\alpha +s -1$. Agent $A_1$ beeps in round $\gamma -1$ and listens in rounds $\gamma -3$, $\gamma -2$ and $\gamma$.
Agent $A_2$ either listens in all rounds $\gamma -3$, $\gamma -2$, $\gamma -1$ and $\gamma$, or beeps in rounds
$\gamma -3$ and $\gamma -1$ and listens in rounds $\gamma -2$ and $\gamma$. In both cases one of the agents hears a beep either in round
$\gamma -3$ or $\gamma -1$. 

{\bf Case 3.} Agents are activated in different rounds.

Without loss of generality let $A_1$ be activated earlier and let $\alpha$ be its activation round. Let $\beta>\alpha$ be the activation round of $A_2$.
Let $s$ be the length of $T_2(L_1)$.

If $\beta=\alpha+1$ or $\beta=\alpha+2$  then $A_2$ hears a beep in round $\alpha +2$. 
If $\beta=\alpha+3$ then $A_2$ hears a beep in round $\alpha +4$ because the first bit of the binary representation of every label $L$ is 1, hence the 5th
bit of $T_2(L_1)$ is 1. 

Next suppose that  $\alpha +4 \leq \beta \leq \alpha +s-8$. Consider three possibilities. 
If $\beta -\alpha$ is odd then agent $A_2$ beeps in round $\beta +2$ and agent $A_1$ listens in this round, hence it hears a beep.
If $\beta -\alpha$ is divisible by 4 then agent $A_2$ listens in round $\beta$ and beeps in round $\beta +2$.
Agent $A_1$  either listens in rounds $\beta$ and $\beta +2$ or beeps in these rounds. Hence one of the agents hears a beep in one of these rounds.
The case of  $\beta -\alpha=4i+2$, for some integer $i$, is slightly more complicated. In this case divide all rounds, starting from round $\alpha$, into segments of size 2.
Round $\beta$ is in the beginning of such a segment. Segments correspond to bits of sequences $T_1(L_1)$ and $T_1(L_2)$.
Let $T_1(L_1)$ be the binary sequence $(01c_1c_1c_2c_2\dots c_kc_k01)$ and consider the final bits 01 of $T_1(L_2)$. 
Let $I$ be the segment corresponding to this bit 0 and let $J$ be the segment corresponding to this bit 1.
There are four possible
situations:

1. The segment $I$ corresponds to some $c_j$ in $T_1(L_1)$ and
the segment $J$ corresponds to $c_{j+1}$.

In this case the second copy of $c_{j+1}$ corresponds to a segment in which agent $A_2$ has already terminated and listens.
If $c_{j+1}=0$ then $J$ corresponds to 0 in $T_1(L_1)$ and to 1 in $T_1(L_2)$. Hence $A_1$ hears a beep in the first round of this segment.
If $c_{j+1}=1$ then in the first round of the segment following $J$ agent $A_1$ beeps and agent $A_2$ listens,
hence it hears a beep.

2. The segment $I$ corresponds to $c_k$ in $T_1(L_1)$.

In this case $I$ corresponds to the second copy of $c_k$ in $T_1(L_1)$, and hence $J$ corresponds to 0 in $T_1(L_1)$ and to 1 in $T_1(L_2)$.
Consequently, agent $A_1$ hear a beep in the first round of this segment.

3. The segment $I$  corresponds to the final 1 in $T_1(L_1)$.

In this case $A_2$ hears a beep in the first round of segment $I$. 

4. In the segment $I$ agent $A_1$ has already terminated and listens.

In this case, in segment $J$ corresponding to 1 in $T_1(L_2)$ agent $A_1$ also listens after termination. Hence, it hears a beep in the first round of this segment. 
 
It follows that, whenever $\alpha +4 \leq \beta \leq \alpha +s-8$, one of the agents hears a beep at the latest in the round following the last round corresponding to a 
bit from $T_2(L_2)$ (or equivalently, at the latest in the round following the last round corresponding to a 
bit from $T_1(L_2)$).

If $\beta=\alpha +s-7$ then agent $A_2$ listens in round $\beta +1$ and beeps in round $\beta +2$, while agent
$A_1$ either listens in both these rounds or beeps in round $\beta +1$ and listens in round $\beta +2$. 
Hence one of the agents hears a beep in one of these rounds.
If $\beta=\alpha +s-6$ then agent $A_2$ beeps in round $\beta +2$ and agent $A_1$ listens in this round, hence it hears a beep. 
If $\beta=\alpha +s-5$ then agent $A_2$ beeps in round $\beta +2$ and agent $A_1$ listens in this round, hence it hears a beep. 
Finally, if $\beta>\alpha +s-5$, then agent $A_2$ beeps in round $\beta +4$ (because the first bit of the binary representation of every label $L$ is 1, hence the 5th
bit of $T_2(L_2)$ is 1)  
and agent $A_1$ listens in this round (as it already concluded processing bits
of $T_2(L_1)$), hence it hears a beep. This concludes the argument in Case 3.

Hence one of the agents must always hear a beep in some round $\xi$. Moreover, our case analysis shows that round $\xi$ is at most $O(\log \ell)$
rounds after the activation of the later agent, where $\ell$ is the smaller label.

Let $\rho$ be the first round in which one of the agents hears a beep. Call this agent $A_3$ and the other agent $A_4$. Notice that $A_3$ could not beep in round $\rho -1$.
Otherwise, agent $A_4$ that was already active in round $\rho -1$ and listened in this round (agents never beep in two consecutive rounds) would hear a beep
in round $\rho -1$, contradicting the definition of round $\rho$. Hence, according to the procedure, agent $A_3$ beeps in round $\rho +1$ and declares round
$\rho +2$ {\em red}. Agent $A_4$ which listens in round $\rho +1$ hears a beep in this round for the first time. Since it beeped in round $(\rho +1)-1$, it declares round
$(\rho +1)+1$ {\em red}. Hence both agents declare the same round $\rho +2$ {\em red}. 

For agent $A_4$, round {\em red} is the next round after it heard a beep for the first time, and this agent sets $role:=waiting$. 
For agent $A_3$ round {\em red} is two rounds
after it heard a beep for the first time, and this agent sets $role:=walking$. 

Since round $\rho$  is at most $O(\log \ell)$
rounds after the activation of the later agent, the same is true for the round declared  {\em red}, which concludes the proof of the lemma.
\end{proof}

We will also use a modified version of Procedure {\tt Beeping-exploration}, described at the beginning of this section, for an integer parameter~$n$.

{\bf Procedure} {\tt Modified-beeping-exploration} $(n)$

Let $EXP(n)$ be the procedure described in Section 2 that permits exploration of all graphs of size at most $n$. Replace each round $r$ of $EXP(n)$ by two consecutive rounds as follows. If in round $r$ of $EXP(n)$ the agent takes port $p$ to move to node $w$, then in the first of the two replacing rounds the agent
takes port $p$ to move to $w$ and beeps, and in the second replacing round it stays at $w$ and listens. 

Hence, in each of the two rounds replacing a round $r$ of $EXP(n)$, the agent is at the same node in Procedure  {\tt Modified-beeping-exploration} $(n)$ as it is in Procedure $EXP(n)$  in round $r$.

Below we give the pseudo-code of the algorithm executed by an agent with label $L$ in a graph of size at most $n$.

\begin{center}
\fbox{
\begin{minipage}{16cm}

{\bf Algorithm} {\tt Fast-bounded-energy-RV-with-detection}\\

Perform Procedure {\tt Symmetry-breaking}\\
{\bf if} $role=waiting$ {\bf then}\\
\hspace*{1cm}stay idle and listen {\bf until} you hear a {\em loud} beep\\
\hspace*{1cm}let $t$ be the round when you first hear a {\em loud} beep  (counted since your wake-up)\\
\hspace*{1cm}beep in round $t+1$, declare rendezvous in round $t+2$, and stop\\
{\bf else}\\
\hspace*{1cm}perform Procedure {\tt Modified-beeping-exploration} $(n)$ starting in round {\em red}\\ 
\hspace*{1cm}(previously declared in the execution of Procedure {\tt Symmetry-breaking})\\
\hspace*{1cm}{\bf until} you hear a {\em loud} beep\\
\hspace*{1cm}let $t$ be the round when you first hear a {\em loud} beep  (counted since your wake-up)\\
\hspace*{1cm}declare rendezvous in round $t+1$, and stop.

\end{minipage}
}
\end{center}

\begin{theorem}\label{ub3}
For any positive integer constant  $n$ there exists a positive integer  $c$, such that Algorithm {\tt Fast-bounded-energy-RV-with-detection} can be executed by $c$-bounded agents in any graph of size at most $n$, in the global beeping model. If two such agents with distinct labels execute this algorithm in such a graph, then they meet and simultaneously declare rendezvous in time $O(\log \ell)$
after the start of the later agent, where $\ell$ is the smaller label.  This time is optimal, even in the two-node graph.
\end{theorem}

\begin{proof}
Let $n$ be a positive integer constant. In order to show the existence of the integer $c$, it is enough to show that the number of steps prescribed by the algorithm depends only on
the constant $n$ (and not on the label $L$ of the agent). The agent moves  only during the execution 
of Procedure  {\tt Modified-beeping-exploration} $(n)$, once every two rounds. Since a full execution of 
Procedure  {\tt Modified-beeping-exploration} $(n)$ takes $2R(n)$ rounds (where $R(n)$ is the execution time of $EXP(n)$)  it is enough to take $c=R(n)$.

We now prove the correctness of the algorithm. By Lemma \ref{breaking}, both agents declare the same round to be {\em red}. In Procedure {\tt Symmetry-breaking},
either both agents heard a {\em loud} beep, or they both heard a {\em soft}
beep, because agents do not move during the execution of this procedure.
If the beep they heard was {\em loud} then they are at the same node in round {\em red} and they correctly  declare rendezvous in this round and stop.
Hence we may assume that the beep they heard was {\em soft}. By Lemma \ref{breaking}, in round {\em red} one of the agents, call it $A_1$, has the variable $role$
previously set to $waiting$, and the other agent, call it $A_2$, has the variable $role$ previously set to $walking$. Agent $A_1$ stays idle forever
at its starting node $v$. Let $\sigma$ be the first round in which agent $A_2$ arrives at node $v$. Such a round exists because $EXP(n)$ explores the entire graph.
Agent $A_2$ beeps in round $\sigma$. Agent $A_1$ hears a {\em loud} beep in this round for the first time. Hence it beeps in round $\sigma +1$ and declares rendezvous in round $\sigma +2$. Agent $A_2$ hears a loud beep for the first time in round $\sigma +1$. Hence it declares rendezvous in round $\sigma +2$. This concludes the proof of correctness.

It remains to estimate the time between the start of the later starting agent and the declaration of rendezvous.
By Lemma \ref{breaking}, the round declared {\em red} in Procedure {\tt Symmetry-breaking}  
is $O(\log \ell)$ rounds after the activation of the later agent, where $\ell$ is the smaller label.
 Since $n$ is constant, $R(n)$ is also constant, and hence the duration of the execution of Procedure  {\tt Modified-beeping-exploration} $(n)$
 is constant. It follows that the time between the activation of the later agent and the declaration of rendezvous is
 $O(\log \ell)$, where $\ell$ is the smaller label. 
 
 In order to prove that time $\Theta(\log \ell)$ is optimal, even in the global beeping model, consider the two-node graph $K_2$.
  For any integer $x>2$ and any algorithm $\cal A$ for rendezvous with detection in the global beeping model,
 working in time $t<(x-1)/2$, we show two labels
 $L_1$ and $L_2$ with binary representations of length $x$, such that this algorithm fails if agents, placed at both nodes of  $K_2$ and activated simultaneously, have these labels. In every round, the behavior of an agent can be described as a pair of bits, the first of which determines if the agent stays or moves in the given round, and the second of which determines if the agent beeps or does not beep in the given round. Hence, in an execution of an algorithm working in time $t$, there are
 $4^t$ possible behaviors of agents. In any execution in which an agent does not hear any beep, the behavior of the agent depends only on its label. There are $2^{x-1}$ labels with binary representations of length $x$
 (each representation starts with 1). Consequently, if $t<(x-1)/2$, then $2t<x-1$, and hence $4^t<2^{x-1}$, which implies that there are at least two labels 
 with binary representations of length $x$, which induce identical behavior of the agents. None of such agents hears any beep,  as they beep in exactly the same rounds, 
and they are at different nodes at all times. Hence the algorithm fails for them. This implies that rendezvous with detection must take time $\Omega(\log \ell)$, where $\ell$ is the smaller label. 
\end{proof}

\section{Conclusion}

We presented three algorithms of rendezvous with detection. The first two of them work
even in the local beeping model: one for unrestricted agents in arbitrary graphs,
and the other for bounded-energy agents in bounded-size graphs. We showed that in the latter case the meeting time of bounded-energy
agents must be  exponentially larger than the best time
of rendezvous with detection of unrestricted agents. More precisely, in order to meet
in bounded-size graphs,  bounded-energy agents must use time polynomial in the larger label, while
unrestricted agents can meet in time polylogarithmic in the smaller label.  
The third algorithm works for bounded-energy agents only in the global beeping model, but it is much faster: it enables such agents to perform rendezvous with detection in bounded-size graphs in time logarithmic in the smaller label, which is optimal.

Rendezvous with detection may be considered as a preprocessing procedure for other important tasks in graphs. One of them is the task of constructing
a map of an unknown graph by an agent. It is well known that this task cannot be accomplished by a single agent operating in a graph, if it cannot mark nodes 
(e.g., a single agent cannot learn the size of an oriented ring).
For the same reason it cannot be accomplished by two non-communicating agents, as they would not be aware of the presence of each other, and thus each of them would act as a single agent. By contrast, our algorithms of rendezvous with detection in the beeping model can serve, with a simple addition,
to achieve map construction by the agents: the algorithm working for arbitrary agents can be used to accomplish this task in arbitrary graphs, 
and the algorithms working for bounded-energy agents can be used to accomplish this task in bounded-size graphs. This addition can be described as follows.
Note that, in all our algorithms, at the time when agents declare rendezvous, symmetry
between them is broken: in the case of algorithms in the local model, one of the agents heard two beeps at the meeting node, and the other agent heard only one beep, and in the case of the
algorithm in the global model, one of the agents has $role$ set to $waiting$ and the other to $walking$. Hence agents can start simultaneously the following procedure in the round after rendezvous declaration. The first agent stays idle
and acts as a stationary token, beeping in every second round, while the second agent silently executes exploration with a stationary token (at the end of which it acquires the map of the graph), cf. e.g., \cite{CDK}, replacing each exploration round by two rounds, in the first of which it moves as prescribed and in the second it stays idle. Beeps of the idle agent allow the circulating silent agent to recognize the token at each visit and complete exploration and
map construction. At the end of the exploration, the second agent is with the first one and can inform it of the end of the exploration by beeping in the last round, in which the first agent is silent (listens).  
Then the roles of the agents may change to allow the previously idle agent to acquire the map in its turn. (Note that an agent
cannot efficiently communicate the already acquired map due to the restrictive communication model.)

\pagebreak

%%%%%%%%%%%%%%%%%%%%%%%%%%%%%%%%%%%%%%%%%%%%%%%%%%%%%%%%%%%
\bibliographystyle{plain}

\begin{thebibliography}{12}
%%%%%%%%%%%%%%%%%%%%%%%%%%%%%%%%%%%%%%%%%%%%%%%%%%%%%%%%%%%
\bibitem{AABCHK}
Y. Afek, N. Alon, Z. Bar-Joseph, A. Cornejo, B. Haeupler, F. Kuhn, Beeping a maximal independent set.
Proc. 25th International Symposium on Distributed Computing (DISC 2011), LNCS 6950, 32-50.

%\bibitem{alpern02a}
%S. Alpern,
%Rendezvous search on labelled networks,
%Naval Research Logistics 49 (2002), 256-274.

\bibitem{alpern02b}
S. Alpern and S. Gal,
The theory of search games and rendezvous.
Int. Series in Operations research and Management Science,
Kluwer Academic Publisher, 2002.




\bibitem{anderson98b}
E. Anderson and S. Fekete,
Two-dimensional rendezvous search,
Operations Research 49 (2001), 107-118.

\bibitem{BIOKM}
D. Baba, T. Izumi, F. Ooshita, H. Kakugawa, T. Masuzawa, Space-optimal rendezvous of mobile agents in asynchronous 
trees. Proc. 17th Int. Colloquium on Structural Information and Comm. Complexity, (SIROCCO 2010), LNCS 
6058, 86-100. 


\bibitem{BCGIL}
E. Bampas, J. Czyzowicz, L. Gasieniec, D. Ilcinkas, A. Labourel, Almost optimal asynchronous rendezvous in infinite multidimensional grids,
Proc. 24th International Symposium on Distributed Computing (DISC 2010), LNCS 6343, 297-311.


\bibitem{CDK}
J. Chalopin, S. Das, A. Kosowski, 
Constructing a map of an anonymous graph: Applications of universal sequences,
Proc. 14th International Conference on Principles of Distributed Systems (OPODIS 2010), 119-134.




\bibitem{CFPS}
M. Cieliebak, P. Flocchini, G. Prencipe, N. Santoro, 
Distributed computing by mobile robots: Gathering, SIAM J. Comput. 41 (2012), 829-879.


\bibitem{CK}
A. Cornejo, F. Kuhn, Deploying wireless networks with beeps,
Proc. 24th International Symposium on Distributed Computing (DISC 2010), LNCS 6343, 148-162.


\bibitem{CKP}
J. Czyzowicz, A. Kosowski, A. Pelc, How to meet when you forget: Log-space rendezvous in arbitrary graphs, Distributed Computing 25 (2012), 165-178. 

%\bibitem{CLP}
%J. Czyzowicz, A. Labourel, A. Pelc, How to meet asynchronously (almost) everywhere, 
%ACM Transactions on Algorithms 8 (2012), article 37. 

%\bibitem{DGKKP}
%G. De Marco, L. Gargano, E. Kranakis, D. Krizanc, A. Pelc, U. Vaccaro, 
%Asynchronous deterministic rendezvous in graphs, 
%Theoretical Computer Science 355 (2006), 315-326.


%\bibitem{DKM}
%B. Degener, B. Kempkes, F. Meyer auf der Heide, 
%A local O(n2) gathering algorithm, 
%Proc. 22nd Ann. ACM Symposium on Parallelism in Algorithms and Architectures (SPAA 2010), 217-223.


\bibitem{DFKP}
A. Dessmark, P. Fraigniaud, D. Kowalski, A. Pelc.
Deterministic rendezvous in graphs.
Algorithmica 46 (2006), 69-96.

\bibitem{DPV}
Y. Dieudonn\'{e}, A. Pelc, V. Villain, How to meet asynchronously at polynomial cost, SIAM J. Comput. 44 (2015), 844-867. 


\bibitem{fpsw}
P. Flocchini, G. Prencipe, N. Santoro, P. Widmayer,
Gathering of asynchronous robots with limited visibility, Theoretical Computer Science 337 (2005), 147-168.

\bibitem{FSVY}
P. Flocchini, N. Santoro, G. Viglietta, M. Yamashita, Rendezvous of two robots with constant memory,
Proc. 20th Int. Colloquium on Structural Information and Comm. Complexity (SIROCCO 2013), LNCS 8179, 189-200.



%\bibitem{FP1}
%P. Fraigniaud, A. Pelc, Deterministic rendezvous in trees with little memory, 
%Proc. 22nd International Symposium on Distributed Computing (DISC 2008),  242-256. 


\bibitem{FP2}
P. Fraigniaud, A. Pelc, Delays induce an exponential memory gap for rendezvous in trees, ACM Transactions on Algorithms 9 (2013), article 17. 

\bibitem{GN}
S. Gilbert, C. Newport, The computational power of beeps, Proc. 29th International Symposium on Distributed Computing (DISC 2015), 31-46.



\bibitem{Ko}
M. Kouck\'{y}, Universal traversal sequences with 
backtracking, Journal of Computer and System Sciences  65 (2002), 717-726. 





\bibitem{KKPM08}
E. Kranakis, D. Krizanc, and P. Morin, 
Randomized Rendez-Vous with Limited Memory,
Proc. 8th Latin American Theoretical Informatics (LATIN 2008),  LNCS 4957, 605-616.

\bibitem{KKSS}
E. Kranakis, D. Krizanc, N. Santoro and C. Sawchuk, 
Mobile agent rendezvous in a ring, 
Proc. 23rd Int. Conf. on Distr. Computing Systems
(ICDCS 2003), 592-599.

\bibitem{MP}
A. Miller and A. Pelc, Time versus cost tradeoffs for deterministic rendezvous in networks,
 Proc. 33rd Annual ACM Symposium on Principles of Distributed Computing (PODC 2014), 282-290.
 
 \bibitem{P}
 A. Pelc, Deterministic rendezvous in networks: A comprehensive survey, Networks 59 (2012), 331-347. 

\bibitem{Re}
O. Reingold, Undirected connectivity in log-space, Journal of the ACM 55 (2008).



\bibitem{TSZ07}
A. Ta-Shma and U. Zwick.
Deterministic rendezvous, treasure hunts and strongly universal exploration sequences.
Proc. 18th ACM-SIAM Symposium on Discrete Algorithms (SODA 2007), 599-608.

\bibitem{YJYLC}
J. Yu, L. Jia, D. Yu,  G. Li, X. Cheng, Minimum connected dominating set construction in wireless networks
               under the beeping model, Proc. IEEE Conference on Computer Communications, (INFOCOM 2015), 972-980.


%%%%%%%%%%%%%%%%%%%%%%%%%%%%%%%%%%%%%%%%%%%%%%%%%%%%%%%%%%%
\end{thebibliography}

%%%%%%%%%%%%%%%%%%%%%%%%%%%%%%%%%%%%%%%%%%%%%%%%%%%%%%%%%%% 

\end{document}